\documentclass[preprint,authoryear,12pt]{elsarticle}

\usepackage{graphics}
\usepackage{pgfplots}
\usepackage{graphicx}
\usepackage{epsfig}

\usepackage{amssymb}

\journal{Statistics and Probability Letters}

\usepackage{graphicx}
\usepackage{psfrag}
\usepackage[active]{srcltx}

\usepackage{amsmath,amsfonts,amssymb,latexsym,epsfig, amsthm}

\newtheorem{prop}{Proposition}
\newtheorem{rem}{Remark}

\newtheorem{assum}{Assumption}

\newcommand{\R}{\mathbb{R}}
\newcommand{\Hs}{\mathcal{H}}

\newcommand{\N}{\mathcal{N}}

\newcommand{\Cov}{\mathcal{C}}

\renewcommand{\epsilon}{\varepsilon}

\renewcommand\phi{\varphi}

\renewcommand{\L}{\ell} 
\newcommand{\E}{\mathrm{E}}

\begin{document}

\begin{frontmatter}



\title{A Stable Manifold MCMC Method for High Dimensions}


\author{Alexandros Beskos}

\address{Dept of Statistics and Applied Probability, National University of Singapore}

\begin{abstract}

We combine two important recent  advancements of MCMC algorithms: first, methods utilizing the intrinsic manifold structure of the parameter space; then, algorithms effective for targets in infinite-dimensions with the critical property that their mixing time is robust to mesh refinement.

\end{abstract}

\begin{keyword}
Manifold MCMC \sep Metropolis-adjusted Langevin algorithm \sep 
Cameron-Martin space, infinite dimensions.


\end{keyword}

\end{frontmatter}



\section{Introduction}

Manifold MCMC methods were introduced in \cite{giro:11} and
were shown to be effective for challenging  target distributions
with complex local-correlation structures. Furthermore, MCMC methods
robust in high dimensions have been recently developed (see e.g.\@ \cite{besk:08}) for
important statistical models giving rise to targets defined as change of measures from
Gaussian laws in infinite dimensions. This work aims to develop new MCMC algorithms
with the objective of joining strengths from the two 
directions of recent methodological progress.
The new  methodology will be illustrated on a model structure involving a diffusion process 
observed with (small) error. This simple example has been selected to clearly showcase the effect of the new method, 
as the infinite-dimensional aspect of the MCMC algorithm will deal with the high-dimensionality of the diffusion path, while 
its manifold aspect will provide a principled mechanism for driving proposed diffusion paths close to the data. 

We will focus on the Manifold Metropolis-adjusted Langevin algorithm (MMALA) 
of \cite{giro:11}.
On the infinite-dimensional side, the methods in, e.g.\@, \cite{besk:08, cott:13} are relevant for targets
$\Pi$  defined as change of measures from Gaussian laws, i.e.:
\begin{equation}
\label{eq:pi}
\frac{d\Pi}{d\tilde{\Pi}}(x) = \exp\{ -\Phi(x)\}\ , \quad  x\in\Hs\  , 
\end{equation}
where $\Hs$ is a separable Hilbert space equipped with the inner product $\langle \cdot, \cdot \rangle$, 
$\tilde{\Pi}= \N(\mu, \Cov)$  a Gaussian distribution on $\Hs$ of mean $\mu$ 
and covariance $\Cov$, and $\Phi:\Hs\mapsto \R$. See e.g.\@ \cite{prat:92} for a treatment of Gaussian laws on general 
Hilbert spaces. The approaches in  \cite{besk:08}  are effective in the case when $\Hs$ is infinite-dimensional
since - upon  finite-dimensional projection and selection of some relevant mesh size -  they provide algorithms of \emph{mesh-free} mixing time, separating themselves from standard MCMC algorithms for which 
mixing time deteriorates with increasing dimension (see e.g. \cite{robe:01}).
The contribution of this paper will be to combine the manifold and infinite-dimensional methods, 
thus deriving new algorithms that could unify the positive computational effects of these two approaches.
The main algorithm developed in this paper will be assigned the label $\infty$-MMALA (the label $\infty$-MALA 
will be used for the infinite-dimensional MALA of \cite{besk:08}; recall that MMALA refers to the manifold method).

The structure of the paper is as folllows.
In Section \ref{sec:algorithm} we develop $\infty$-MMALA and state conditions under which it is well-defined in infinite dimensions. In Section \ref{sec:example} we  show numerical results from applying $\infty$-MMALA on a diffusion-driven model. In Section \ref{sec:directions} we discuss further directions.

\section{$\infty$-MMALA: Manifold MALA in Infinite-Dimensions}
\label{sec:algorithm}
We  focus mainly  on the practicalities of the derivation of the algorithm and avoid  technicalities.
Therefore, we will not discuss here the notion of differentiation (denoted by $\nabla$) in general Hilbert spaces or the well-posedness of the Langevin stochastic differential equation (SDE) below or its manifold version on arbitrary Hilbert spaces. 
The reader could simply  assume that the development of the algorithm happens on some $N$-dimensional projection 
of the infinite-dimensional target $\Pi$, for some large $N\ge 1$, so that the state-space is the Euclidean $\R^{N}$.
Mathematical rigor will be applied  when defining the final algorithm (involving 
the easier to handle time-discretised SDE dynamics). So, $\Pi$ is used interchangeably below to denote both the 
infinite-dimensional target and the $N$-dimensional projection; a similar convention is applied for other related notions, e.g.~for
$\tilde{\Pi}$. Also, from the definition of $\Pi$ in (\ref{eq:pi}) we  have  (in a formal sense, in the case of general Hilbert spaces):
\begin{equation*}
\Pi(x) \propto \exp\{\L(x)\} = \exp\big\{-\Phi(x) - \tfrac{1}{2}\langle x-\mu, L(x-\mu) \rangle\big\}\  ,
\end{equation*}
where we have set
$L = \Cov^{-1}$.
%

\subsection{MMALA and $\infty$-MALA Algorithms}

MMALA utilizes the dynamics of the Langevin SDE on the manifold space generated by 
a chosen metric tensor $G(x)$. Its expression is as
follows: 
\begin{align}
\label{eq:msde}
d {x} = \tfrac{1}{2}\,\tilde{\nabla}\L({x})dt + d\tilde{b}\  ,
\end{align}
with $\tilde{\nabla} = G^{-1}({x})\nabla$ corresponding to 
differentiation along the manifold and $d\tilde{b}$ denoting infinitesimal increments of a Brownian 
motion on the manifold space. In agreement with the comments above, for all practical purposes one 
can assume that the state-space is $x\in \R^N$, so that  $G^{-1}(x)$ is assumed to be a symmetric positive-definite matrix 
in $\R^{N\times N}$. 
Using the analytical expressions from \cite{giro:11} we can 
 equivalently re-express (\ref{eq:msde}) in terms of the following more familiar SDE on Euclidean space: 
\begin{align}
d{x} = G({x})^{-1}\,\big\{ \tfrac{1}{2}\,\nabla\L({x}) + 
\tfrac{1}{2}\nabla \log|G({x})| + \nabla  \big\}\,dt 
+ G({x})^{-1/2}\,db \  ,
\label{eq:sde}
\end{align}
with $db$ now denoting increments of  standard Brownian motion and $|G(x)|$ is the determinant of $G(x)$.
The Langevin SDE (\ref{eq:sde}) will now be time-discretised to provide a proposal in a Metropolis-Hastings framework.
The two algorithms, MMALA and $\infty$-MALA are now specified as follows:
\begin{itemize}
\item MMALA: it applies the standard Euler finite-difference scheme to time-discretise the dynamics (\ref{eq:sde}).
For current position $x$, this will provide a proposed transition to, say, $x'$, accepted or rejected according to the related Metropolis-Hastings ratio. This algorithm will typically 
not be well-defined in infinite-dimensions.
\item $\infty$-MALA: in this case $G(x)\equiv L$, so the drift function in (\ref{eq:sde}) becomes:
\begin{equation*}
G^{-1}(x)\,\tfrac{1}{2}\nabla \L(x)\equiv \tfrac{1}{2}\,\big\{  -\Cov\,\nabla \Phi(x) - (x-\mu) \big\} \ .
\end{equation*}
$\infty$-MALA employs a semi-implicit discretisation scheme, where the linear term $x$ in the drift is replaced  by
with $\frac{x'+x}{2}$ when applying finite differences. Solving w.r.t.\@ $x'$ delivers a proposal of positive acceptance 
probability even in infinite dimensions  (under conditions). 
This is because due to the semi-implicit scheme, and under weak conditions on $\Phi$, $\tilde{\Pi}$, 
the distributions of $(x,x')$ and $(x',x)$ are absolutely continuous w.r.t.\@ 
each other,  thus allowing for a non-zero Metropolis-Hastings ratio. We will discuss this also in the sequel, as $\infty$-MMALA
will require some of the tools used for the development $\infty$-MALA.
\end{itemize}

\subsection{Time-Discretisation of Langevin Dynamics for $\infty$-MMALA}

We will in fact opt for a simplified version of the dynamics in (\ref{eq:sde}), with the objective of combining designated moves with computational efficiency. 
As already noted in \cite{giro:11}
most of the strength of the manifold method is  captured by the dynamics that only involve the term 
$G^{-1}(x)\,\frac{1}{2}\,\L(x)$ in the drift function. Calculation of the removed Christofell symbols is
typically  expensive (of the order of $\mathcal{O}(N^3)$)
and could eradicate in the balance their effect on improved mixing.


Guided by the semi-implicit idea behind $\infty$-MALA, we  introduce a time-discretisation scheme 
which  possesses the critical property of giving rise to a well-posed algorithm in infinite-dimensions. 
The scheme is as follows (we add/subtract $G(x)x$ in the drift and apply
an implicit scheme in one term):
%
%
\begin{align}
x' - x = \tfrac{1}{2}\,G({x})^{-1}\,\big\{ -G(x)\,&\tfrac{x'+x}{2} + G(x)\,x + \nabla \L(x) \big\}\,h \label{eq:implicit} \\  &+ \sqrt{h}\,\N(0,G(x)^{-1})\ , \nonumber
\end{align}
for a step-size $h>0$,
which can equivalently be written as:
\begin{align}
\label{eq:proposal}
x' = \tfrac{1-h/4}{1+h/4}\,x + \tfrac{h/2}{1+h/4}\,S(x) + \tfrac{\sqrt{h}}{1+h/4}\,\N(0,G(x)^{-1})\  ,
\end{align}
where we have defined:
\begin{equation*}
S(x) = - G(x)^{-1}\{\, \nabla \Phi(x) - (G(x)-L)x - L\mu\, \}\  .
\end{equation*}
Notice that the choice $G(x)=L$ will deliver $\infty$-MALA. Equation (\ref{eq:proposal}) provides
the proposal for $\infty$-MMALA upon which we will apply the Metropolis-Hasting acceptance rule. In Section \ref{sub:acc} below we give the expression for the acceptance probability 
of proposal (\ref{eq:proposal}) on the infinite-dimensional space $\Hs$.

\subsection{Choice of Metric Tensor $G(x)$}
Following \cite{giro:11}, an often effective  approach
is to choose the metric tensor as the expected Fisher information (we write 
$\Phi(x)=\Phi(x;y)$ to emphasize dependence of $\Phi$ on some data $Y=y$): 
\begin{align}
\label{eq:tensor}
-\E_{Y|{x}}\nabla^{2}\L({x}) &= \E_{Y|{x}}\nabla^2_x\Phi({x};Y) + L \\
&= \E_{Y|{x}}\big[\,\nabla_{{x}} \Phi({x};Y)\,\{\nabla_{x} \Phi({x};Y)\}^{\top}\,\big] + 
L \  .\nonumber
\end{align} 
Used in the context of high-dimensional $x\in \R^{N}$, this  can sometimes lead to prohibitive computations
as an order of $N$. Thus, for a given class of target distributions one could try to balance improved
mixing due to a good choice of $G(x)$ with computational considerations, and maybe opt for a convenient proxy of the expected Fisher information 
(or the observed Fisher information, or other tensor understood to be appropriate in a given scenario).

\subsection{Diversion on Metropolis-Hastings in General State-Spaces}
Our objective is to show that the acceptance ratio is non-trivial when working 
on an infinite-dimensional Hilbert space $\Hs$. 
We denote by $Q(x,dx')$ the transition probability measure corresponding to $\infty$-MMALA proposal (\ref{eq:proposal}). 
As shown in \cite{besk:08} for $\infty$-MALA, a deviation 
from standard proofs of invariance for Metropolis-Hastings kernels is that there is typically no common 
dominating measure for the probability measures $Q(x,dx')$ over all $x\in \Hs$. So, one has to resort to a
generalised definition of the Metropolis-Hastings ratio in \cite{tier:98}.  
Following \cite{tier:98}, one has to seek for conditions so that - for $x\sim \Pi(dx)$, in stationarity - 
the laws of $(x,x')$ and $(x',x)$ are absolutely continuous w.r.t.\@ each other (we use the symbol 
$`\simeq'$ to denote such a relation between probability laws); then, their Radon-Nikodym 
derivative provides the Metropolis-Hastings acceptance ratio.

More analytically, we define the bivariate probability measure on $\Hs\times \Hs$:
\begin{equation*}
\mu(dx,dx') = \Pi(dx)Q(x,dx')\ . 
\end{equation*}
and the corresponding symmetric measure $\mu^{\top}(dx,dx'):=\mu(dx',dx)$. 
Following \cite{tier:98}, if $\mu\simeq \mu^{\top}$ then 
the Metropolis-Hastings acceptance probability is non-trivial and equal to: 
\begin{equation}
\label{eq:acc}
1 \wedge \frac{d\mu^{\top}}{d\mu}(x,x')\ . 
\end{equation}
Thus, we will specify conditions under which $\mu\simeq \mu^{\top}$ and find $(d\mu^{\top}/d\mu)(x,x')$.

\subsection{Proof of Well-Posedness of $\infty$-MMALA}
\label{sub:acc}

The derivation below has connections with the one in \cite{besk:08} for $\infty$-MALA.
To demonstrate well-posedness of $\infty$-MMALA in infinite dimensions we need some assumptions.
\begin{assum}
\label{ass:1}
$\tilde{\Pi}$-a.s.\@ in $x$, we have $\N(0,G(x)^{-1})\simeq \N(0,L^{-1})$, 
with:
\begin{align*}
\kappa(v; x) := \frac{d\N(0,G(x)^{-1})}{d\N(0,L^{-1})}(v)\ ,\quad x,v\in\Hs\ . 
\end{align*}
\end{assum}
\begin{assum}
\label{ass:2}
$\tilde{\Pi}$-a.s.\@ in $x$, quantity $S(x)$ is an element of the Cameron-Martin space of $\N(0,\Cov)$, 
that is $S(x)\in \mathrm{Im}\,\Cov^{1/2}$.
\end{assum}
%
$\mathrm{Im}\,\Cov^{1/2}$ denotes the image space of $\Cov^{1/2}$.
The Cameron-Martin space of $\N(0,\Cov)$ is comprised of all elements of $\Hs$ 
that preserve absolute continuity of $\N(0,\Cov)$ when translating it. In particular, 
we have the following result. 
%
\begin{prop}
\label{eq:RD}
Consider $\N(0,\Sigma)$ on $\Hs$. If
$R(x) = x + \Sigma^{1/2}x_0$
for a constant  $x_0\in \Hs$, then  $\N(0,\Sigma)\simeq N(0, \Sigma)\circ R^{-1}$
 with density:
\begin{align*}
\frac{d\N(0,\Sigma)\circ R^{-1}}{d\N(0,\Sigma)}(x) = \exp\,\big\{\, \langle x_0, \Sigma^{-1/2}x \rangle - 
\tfrac{1}{2}|x_0|^2  \,\big\}\  .
\end{align*}
\end{prop}
This is Theorem 2.21 of \cite{prat:92}.
Notice that we can rewrite the proposal (\ref{eq:proposal}) as:
\begin{align}
\label{eq:Ql}
Q(x,dx')\,:\quad  
x'=\tfrac{1-h/4}{1+h/4}\,x + \tfrac{\sqrt{h}}{1+h/4}\, \N(\tfrac{\sqrt{h}}{2}S(x),G(x)^{-1})\ . 
\end{align}
Let $\tilde{Q}(x,dx')$ 
denote the transition probability law for the update:
\begin{align}
\label{eq:Q}
\tilde{Q}(x,dx')\,:\quad  
x'=\tfrac{1-h/4}{1+h/4}\,x + \tfrac{\sqrt{h}}{1+h/4}\, \N(0,L^{-1})\ . 
\end{align}
We define the reference bivariate Gaussian measure: 
\begin{align*}
\tilde{\mu}(dx, dx') =\tilde{\Pi}(dx)\tilde{Q}(x,dx')\  . 
\end{align*}
It is easy to check that $\tilde{\mu}(dx,dx')$ is symmetric (see \cite{besk:08} for details; this is because the sum of the squares of the scalars in front of $x$ and $\N(0,L^{-1})$ in 
(\ref{eq:Q}) is unit), so that 
$\tilde{\mu}(dx,dx') = \tilde{\mu}^{\top}(dx,dx'):= \tilde{\mu}(dx',dx)$. 

\begin{prop}
\label{prop:N}
Under Assumptions \ref{ass:1} and \ref{ass:2} above,  $\tilde{\Pi}$-a.s.\@ in $x$ we have that 
$\N(\tfrac{\sqrt{h}}{2}S(x),G^{-1}(x))\simeq \N(0,L^{-1})$  with density:
\begin{align*}
&\lambda(v;x)=\frac{d\N(\frac{\sqrt{h}}{2}S(x),G^{-1}(x))}{d\N(0,L^{-1})}(v) =   \\[0.2cm]
&\qquad \quad \exp\Big\{\, \tfrac{\sqrt{h}}{2}\,\big\langle G^{1/2}(x)S(x), G(x)^{1/2}\,v \big\rangle - 
\tfrac{h}{8}\,\big|G(x)^{1/2}S(x)\big|^2 \,\Big\}\,\times \kappa(v;x)\  .
\end{align*}
\end{prop}
\begin{proof}
We use the chain rule:
\begin{align}
\label{eq:chain}
\tfrac{d\N(\frac{\sqrt{h}}{2}S(x),G^{-1}(x))}{d\N(0,L^{-1})}(v) = 
\tfrac{d\N(\frac{\sqrt{h}}{2}S(x),G^{-1}(x))}{d\N(0,G(x)^{-1})}(v) \times \tfrac{d\N(0,G(x)^{-1})}{d\N(0,L^{-1})}(v)\ . 
\end{align}
The last density  is found via Assumption \ref{ass:1}.
For the other, we use the fact that Assumption \ref{ass:1} implies that operators $L^{-1}$, $G(x)^{-1}$ 
have the same Cameron-Martin space (see Feldman-Hajek theorem in \cite{prat:92}), so that applying Proposition \ref{eq:RD} for $\Sigma\equiv G(x)^{-1/2}$, $x_0 \equiv (\sqrt{h}/2)\,G(x)^{1/2}S(x)$ (guaranteed to be a proper element of $\Hs$ 
due to having $S(x)\in \mathrm{Im}\,G(x)^{-1/2}$) will give that the first density on the RHS of (\ref{eq:chain}) 
is equal to:
\begin{equation*}
\quad \exp\Big\{\, \tfrac{\sqrt{h}}{2}\,\big\langle G^{1/2}(x)S(x), G(x)^{1/2}\,v \big\rangle - 
\tfrac{h}{8}\,\big|G(x)^{1/2}S(x)\big|^2 \,\Big\}\ .
\end{equation*}
The proof is now complete.
%
%
%
\end{proof}
From Proposition \ref{prop:N}, eqs (\ref{eq:Ql}), (\ref{eq:Q}) imply directly that $\tilde{\Pi}$-a.s.\@ in 
$x$ we have $Q(x,dx')\simeq \tilde{Q}(x,dx')$, thus also  $\mu(dx,dx')\simeq \tilde{\mu}(dx,dx')$.  
This essentially completes the well-posedness of $\infty$-MMALA as - due to the symmetricity of $\tilde{\mu}(dx,dx')$ -
it  implies that $\mu(dx,dx') \simeq \mu(dx,dx')$. Indeed, we can find the density required in the acceptance 
probability (\ref{eq:acc}) as follows. First, we find:
\begin{align*}
 \frac{d\mu}{d\tilde{\mu}}(x,x') = \frac{d\Pi}{d\tilde{\Pi}}(x)\frac{dQ}{d\tilde{Q}}(x,x') = \exp\{-\Phi(x)\}\,
\lambda(\rho^{-1}(x';x);x)\ , 
\end{align*}
where we denote by $\rho^{-1}(\cdot;x)$ the inverse of the 1-1 mapping:
\begin{equation*}
v \mapsto \rho(v;x) = \tfrac{1-h/4}{1+h/4}\,x + \tfrac{\sqrt{h}}{1+h/4}\,v \ . 
\end{equation*}
From the definition of symmetric measures, we have that  $(d\mu^{\top}/d\tilde{\mu}^{\top})(x,x') =(d\mu/d\tilde{\mu})(x',x) $, thus we finally have that:
\begin{equation}
\label{eq:acce}
\frac{d\mu^{\top}}{d\mu}(x,x') = \frac{(d\mu/d\tilde{\mu})(x',x)}{(d\mu/d\tilde{\mu})(x,x')} = 
\frac{ \exp\{-\Phi(x')\}\,
\lambda(\rho^{-1}(x;x');x')}{ \exp\{-\Phi(x)\}\,
\lambda(\rho^{-1}(x';x);x)} \ .
\end{equation}
Thus, we have proven that $\mu(dx,dx')\simeq\mu^{\top}(dx,dx')$ 
with the above  density. The complete $\infty$-MMALA algorithm can now be summarized
in Table \ref{tab:MALA}.
\begin{table}[!h]
\begin{flushleft}
\medskip
\hrule
\medskip
{\itshape $\infty$-MMALA:}
{\itshape
\begin{enumerate}
\item[(i)] Start with an initial value $x^{(0)}$ from $\N(\mu,L^{-1})$, or another Gaussian law 
absolutely continuous w.r.t.\@ the target $\Pi$ and set $k=0$. \vspace{0.2cm}

\item[(ii)] Given current $x=x^{(k)}$, propose the transition: 
\begin{align*}
x' = \tfrac{1-h/4}{1+h/4}\,x + \tfrac{h/2}{1+h/4}\,S(x) + \tfrac{\sqrt{h}}{1+h/4}\,\N(0,G(x)^{-1})\ , 
\end{align*}
%
%
Set $x^{(k+1)}=x'$ with probability $1\wedge (d\mu^{\top}/d\mu)(x,x')$ 
for $(d\mu^{\top}/d\mu)(x,x')$ as specified in (\ref{eq:acce}), otherwise set $x^{(k+1)}=x$.

%
\item[(iii)] Set $k \to k+1$ and go to (ii).

\end{enumerate}}
\hrule
\end{flushleft}
\vspace{-0.6cm}
\caption{Definition of $\infty$-MMALA.}
\label{tab:MALA}
\end{table}

%

\begin{rem}
We note that earlier works  (e.g.\@ the recent \cite{cott:13}) have looked at 
$\infty$-MALA for a \emph{constant} metric tensor $G(x)$; \cite{cott:13} use $\infty$-MALA
corresponding to the choice $G(x)=L$ for the algorithm described here.
The extension to a non-constant metric tensor is non-trivial and  involved: i) the development of the discretisation 
scheme in (\ref{eq:implicit}) which is not an apparent generalisation of the scheme for a constant metric tensor 
of the earlier works; ii) the analytical justification of the well-posedness in infinite-dimensions of the new algorithm.
\end{rem}

\section{Illustrative Application}
\label{sec:example}
We consider an SDE observed with small error. That is, we have:
\begin{equation}
\label{eq:sd}
dx_t  = a(x_t)dt + dw_t\ , \quad x_0=x^*\in\R\ , 
\end{equation}
where $w$ denotes standard Brownian motion on $\R$, with data points:
\begin{equation*}
y_i = f(x_{t_i}) + \N(0,\sigma^2)\ , \quad 1\le i\le n \  ,
\end{equation*}
for  times $t_1 < \cdots <t_n=T$, a drift  $a:\R\mapsto \R$ and a mapping $f:\R\mapsto \R$.
The target  here is the posterior law $\Pi(dx)$ of the path $x=\{x_t;\,t\in[0,T]\}$ given 
$y=\{y_1,\ldots, y_n\}$, which is of the general form in (\ref{eq:pi}) with:
%
%
%
\begin{align}
\Phi(x) = \Phi_{\sigma}&(x) - \Phi_{b}(x) \ ; \quad 
\Phi_{\sigma}(x) :=  \sum_{i=1}^{n}\frac{(y_i - f(x_{t_i}))^2}{2\sigma^2}\ ; \nonumber  \\
\Phi_{a}(x) &:=  -\int_{0}^{T}a(x_s)dx_s + \tfrac{1}{2}\int_{0}^{T}a^2(x_s)ds  \ . \label{eq:Phi}
\end{align}
$\tilde{\Pi}$ here is the law of a Brownian motion on $[0,T]$ started at $x^*$. 
Also, here $\Hs=L^2([0,T],\R)$, i.e.\@ the space of squared-integrable paths, equipped with the corresponding inner product.
%
There are two main challenges for MCMC algorithms attempting to sample from $\Pi$.

(a) High-Dimensionality of state space. In theory, the  state space is infinite dimensional. In practice, one 
will typically select a large $N\ge 1$ and apply finite-differences to obtain a vector 
$(x_{1},x_{2},\ldots,x_{N})\in \R^{N}$ corresponding to times $\delta, 2\delta, \ldots, N\delta$ 
for mesh-size $\delta=T/N$. 
$\infty$-MALA will be stable under mesh-refinement: the computational cost per step will increase as $\mathcal{O}(N)$, but the mixing time will be $\mathcal{O}(1)$.  MMALA of \cite{giro:11} will deteriorate with decreasing mesh-size $\delta$. 
The work e.g.\@ in \cite{robe:01} suggests that one has to choose $h=\mathcal{O}(N^{-1/3})$ to control the acceptance probability, thus giving a mixing time of
$\mathcal{O}(N^{1/3})$.

(b) Complex a-posteriori covariance structure. 
$\infty$-MALA will deteriorate for decreasing $\sigma>0$
as target $\Pi$ gets distanced from $\tilde{\Pi}$, but the proposal generation mechanism still uses a covariance matrix $G^{-1}(x) \equiv \Cov$ that does not adjust to the complex covariance structure of 
the posterior characterised by small marginal variances at the times of the data and larger ones further from those.  
In contrast to MALA, MMALA accommodates 
for the complex a-posteriori covariance structure of the $x$-path.
%
The newly developed $\infty$-MMALA turns out to be robust both in increasing $N$ and decreasing $\sigma$.

\subsection{Algorithmic Specification and Numerical Results}

As we are interested mainly in the effect of $\sigma$ at the properties of the algorithm, we will obtain the metric 
tensor  $G(x)$ by applying the expected information idea in (\ref{eq:tensor}) only upon $\Phi_{\sigma}$ in (\ref{eq:Phi}) - 
this also guarantees positive-definiteness for the induced $G(x)$.
Thus, we have:
\begin{equation*}
G(x) = \mathrm{diag}_N\Big\{\,\sum_{i=1}^{n}\mathbb{I}_{t_i=j\delta}\,\tfrac{\{f'(x_{j})\}^2}{\sigma^2}\,,\,\,1\le j\le N\,\Big\} + L\ , 
\end{equation*}
for the $N\times N$ tridiagonal covariance matrix of the Brownian motion:
\begin{equation*}
L = 
\left( \begin{array}{ccccc} 
2 & -1 & 0 & 0 & \cdots \\
-1 & 2 & -1 & 0  & \cdots \\
0 & -1 & 2 & -1 & \cdots \\
\vdots & \vdots & \vdots & \vdots & \vdots \\
0  & \cdots & 0 & -1 & 1 
\end{array}
\right) / \delta  \  .
\end{equation*}
%
Critically for the computational properties of $\infty$-MMALA, due to $G(x)$ being tri-diagonal the cost per step is $\mathcal{O}(N)$ (the same holds for MMALA).

\begin{rem}
Regarding Assumptions \ref{ass:1},  \ref{ass:2} here, note that
$\N(0,G(x)^{-1})$ corresponds to a Brownian motion~$\N(0,L^{-1})$ given finite observations with error,
thus clearly $\N(0,G(x)^{-1})\simeq \N(0,L^{-1})$. Then, $\mathrm{Im}\,\Cov^{1/2}$ consists of paths with $x_0=0$ whose $1^{st}$ derivative (in a weak sense) is in $L^2([0,T],\R)$, see e.g. \cite{besk:13}.  We do not present a proof here that $S(x)\in \mathrm{Im}\,\Cov^{1/2}$, but
we mention that in our runs the paths $S(x)$ over the various $x$'s appear to be everywhere 
differentiable apart from the instances of the data where they are only continuous, thus
 everywhere differentiable in the weak sense. 
\end{rem}



We applied $\infty$-MMALA in the following scenario:
\begin{align*}
a(x) = 4-x \  ,\,\,\, x^{*}=2\ ,\,\,\, t_i = i ,\,\,\, n=100\ , \,\,\,   f(x)=x^{3/2}\ , \,\,\, \sigma^2=0.1 \ .
\end{align*}
We used the standard Euler  scheme to discretise $x$ with mesh-size $\delta=10^{-2}$, so that $N=n\delta^{-1}=10^4$. The algorithm was initiated at a path with $x_{t_i}=2$, for $1\le i\le n$, 
with Brownian bridges connecting these points; this position is very far from the center of the target (notice that the $n=100$ `true'
 $x_{t_i}$'s will be scattered around $4$ due to the choice of drift, and $\sigma^2=0.1$).

We used step-size $h=1.0$,  giving average acceptance probability of $82\%$ 
(this changed to $80\%$ when trying $\delta'=\delta/2$, in an empirical 
manifestation of the mesh-free mixing time of $\infty$-MMALA). 
Fig.\ref{fig:traceplot} shows two traceplots for $\infty$-MMALA over $2,000$ iterations,
corresponding to the position of the path $x$ at times $t=37$, $t=36.5$. Notice that the y-axes are on the same 
scale, so the algorithm adjusts automatically to the different sizes of the marginal posteriors (at $t=37$ there is an observation, thus a lot of information about $x_t$). Even if started far from stationarity, 
the algorithm converges almost instantaneously to the target distribution.

\begin{figure}[!h]
\vspace{-2cm}
\begin{centering} 
\hspace{-2.8cm}\includegraphics[scale=0.61]{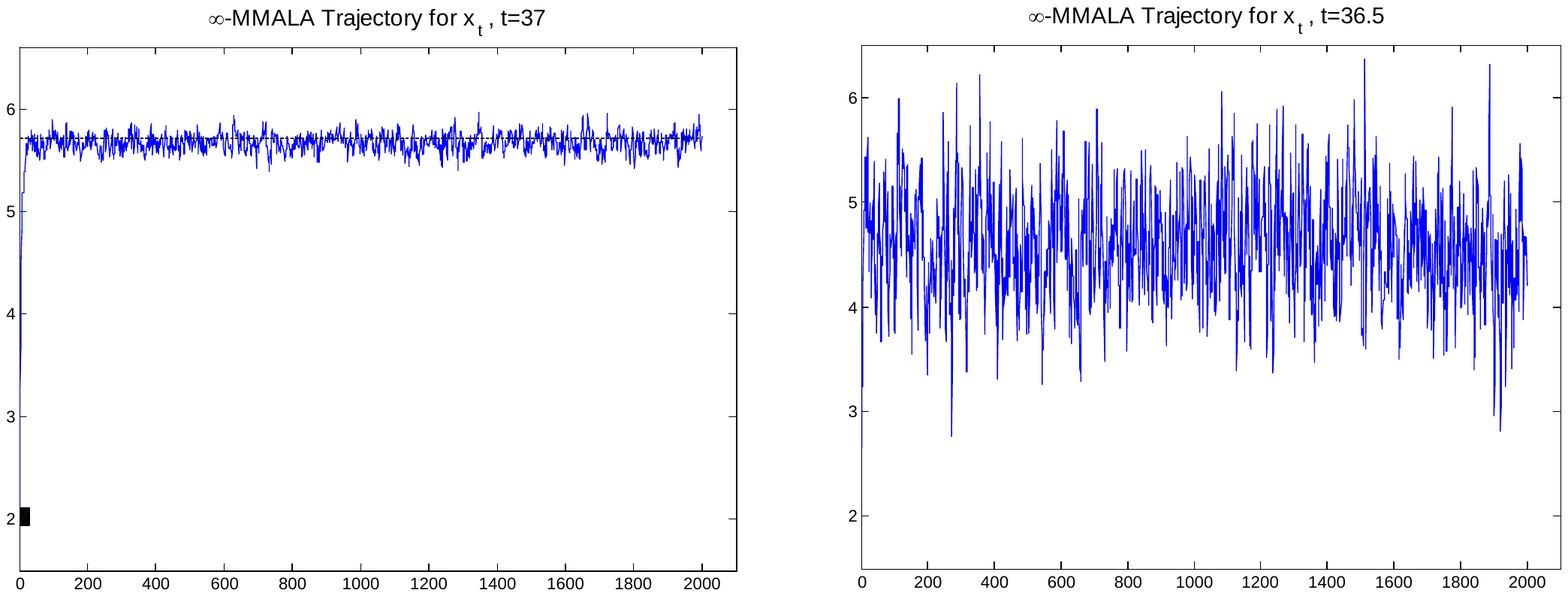} \qquad \qquad 
\end{centering}
\vspace{-7.2cm}
\caption{Traceplots over 2,000 $\infty$-MMALA steps. The left panel corresponds to $x_t$ for $t=37$, 
the right panel to $t=36.5$. The horizontal line on the left panel indicates the value of $f^{-1}(y_{37})$; 
the black rectangle at position 2 of the y-axis highlights the initial position.}
\label{fig:traceplot}
\end{figure}

For comparison, we applied $\infty$-MALA, MMALA in the same setting. $\infty$-MALA was extremely poor, 
as we had to use $h=10^{-5}$ to get acceptance ratio of $53\%$, as the algorithm 
does not adjust to the different (a-posteriori) scales in the target,  and needed a very small $h$ to stay close to the data.
For MMALA, we had to use a step-size of $h = 0.1$ to get average acceptance probability of $69\%$ (reduced to $55\%$ when trying $\delta'=\delta/2$), thus it is much less 
efficient that $\infty$-MMALA (execution times for both MMALA and $\infty$-MMALA were about 40s using Matlab on a
standard  PC). Fig.\ref{fig:QV} highlights a  consequence of the fundamental structural difference between $\infty$-MMALA and MMALA:
 we took an initial path pinned at the data, so $x_{t_i}=f^{-1}(y_i)$ for $1\le i \le n$ with Brownian bridges in-between, 
and run 1,000 iterations of $\infty$-MMALA and MMALA with $h=1.0$. The plots in Fig.\ref{fig:QV} show the 
estimated Quadratic Variation (QVe) of all 1,000 \emph{proposed} paths for both algorithms. Recall here
that  $QVe=\sum_{j=1}^{N}(x_{j}-x_{j-1})^2$, and this quantity 
will converge to $T=100$ as $\delta\rightarrow 0$. As $\infty$-MMALA is well-defined in infinite-dimensions, 
the estimated QV of the path is very close to the limiting one for $N=\infty$ (the acceptance ratio was $81\%$). 
In contrast, MMALA gave QVe's
wide off the mark; not surprisingly, all 1,000 proposed paths were rejected.



\begin{figure}[!h]
\begin{centering} 
\hspace{-2.7cm}\includegraphics[scale=0.6]{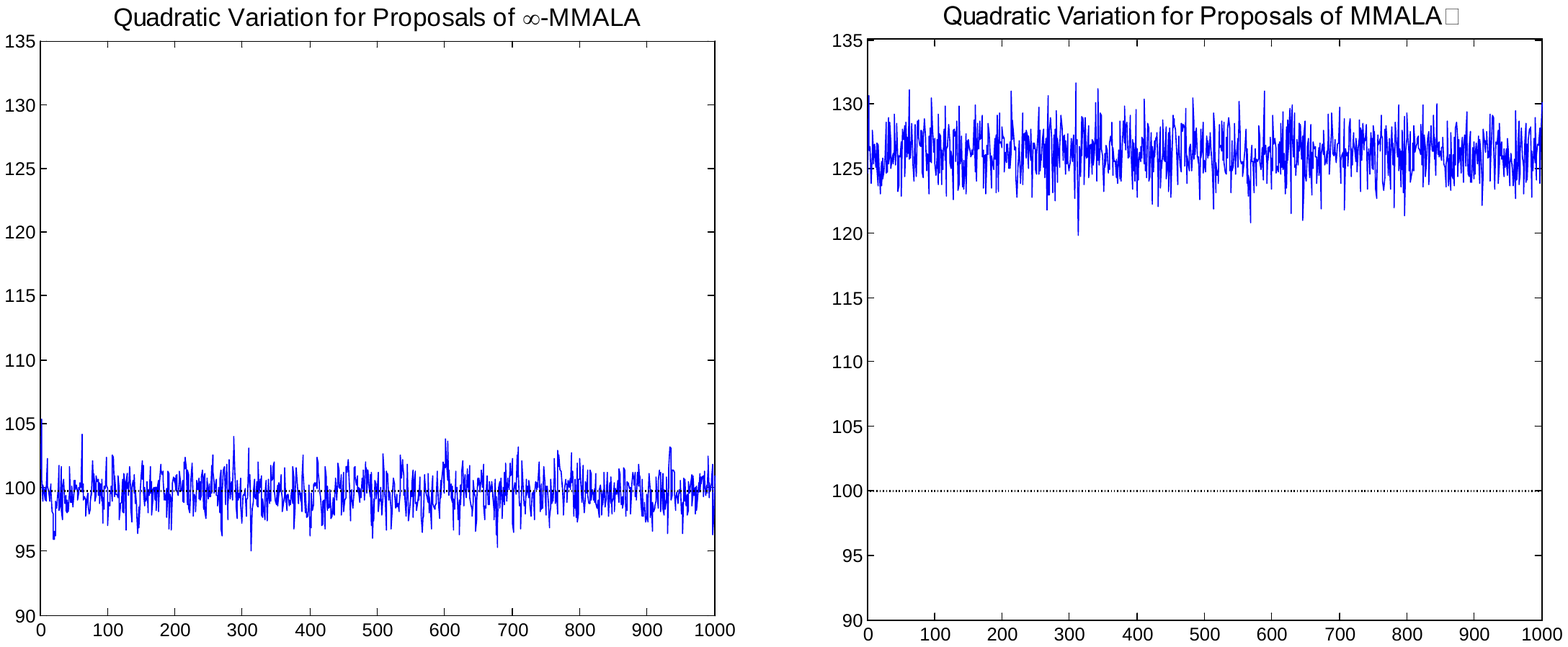} \qquad \qquad 
\end{centering}
\vspace{-6.3cm}
\caption{Estimated QV of the proposed $x$'s over 1,000 iterations of $\infty$-MMALA (left) and MMALA (right) 
both using step-size $h=1.0$. 
The horizontal lines highlight the limiting QV (equal to $T=100$) in infinite dimensions. MMALA proposes $x$'s 
with very wrong QV - so all of them got rejected; $\infty$-MMALA is well-defined in infinite-dimensions, thus 
the estimated QV of the proposed $x$'s is very close to the limiting one (acceptance rate $81\%$).}
\label{fig:QV}
\end{figure}

\section{Conclusions and Further Directions}
\label{sec:directions}

We presented a first attempt at merging in a principled manner
recently developed manifold and infinite-dimensional MCMC algorithms. 
A simple SDE-model served well as an example 
where the new method indeed combines the benefits of the two directions.
We aim to further develop this line of research and clarify the potential of new algorithms 
in important classes of applications. Some further investigations are summarised below.

\emph{Algorithmic Development:} We aim to develop a Hybrid Monte-Carlo (HMC) version of $\infty$-MMALA. Also, 
other approaches for merging infinite-dimensional algorithms with manifold ones could be more appropriate in 
particular settings. In data assimilation,
recent works have set-up a Bayesian framework in infinite-dimensions (see e.g.\@ \cite{stua:10} and the references therein) 
where information for important parameters of interest, such as the 
permeability field for sub-surface flow models, or the initial condition for fluid velocity dynamics,
is expressed in the form of a posterior distribution defined as a change of measure from a Gaussian prior. 
$\infty$-MALA (or a random-walk version of it) has turned out to be overly costly in this context (see e.g.\@ \cite{stua:10}) as the posterior could be characterised by far more complex correlation structure than the prior. 
It seems very natural to develop an $\infty$-MMALA in this context, by applying MMALA  
at the low frequency components of the unknown parameters that are mostly informed by the data, 
and $\infty$-MALA for the high-frequency ones that are mainly determined by the Gaussian prior. 
A similar algorithmic construction could give critical computational advantages in the class of models of SDEs
driven by fractional Brownian motion (fBM), where recent attempts (see e.g.\@ \cite{kalo:14} and the references therein) to apply MCMC have to deal with the existence 
of complex correlation structures among model parameters together with a high-dimensional latent driving fBM. 
Such ideas can be relevant also for Bayesian non-parametric density estimation, e.g.\@ the logistic Gaussian process prior (see e.g.~\cite{tokd:07,cott:13}).

\emph{Weakening Assumptions:} Assumptions \ref{ass:1}, \ref{ass:2} 
seem stronger than needed for the well-posedness of $\infty$-MMALA. In our example model for instance 
we indeed have $\N(0,G(x)^{-1})\simeq \N(0,L^{-1})$ for any $\sigma>0$, but not in the limit when $\sigma=0$, 
this maybe suggesting (falsely, from our experiments)  that the algorithm may deteriorate as $\sigma\rightarrow 0$. 
The resolution is that our assumptions and proof of well-posedness could involve some more `appropriate' Gaussian 
measure instead of $\N(0,L^{-1})$, and in particular one for which its density w.r.t.\@ $\N(0,G(x)^{-1})$ will not be trivial as 
$\sigma\rightarrow 0$. Such considerations are relevant beyond our example model.

\section*{Acknowledgments} I thank the anonymous referee and the Associate Editor for useful suggestions 
that have improved the content of the paper.
\label{sec:directions}

\section*{References}

\bibliographystyle{elsarticle-harv}

\end{document}